\documentclass[prodmode,letterpaper]{acmsmall} 
\usepackage{setspace}
\usepackage{algpseudocode}
\usepackage{algorithm}
\usepackage{float}
\usepackage{amsmath}
\usepackage[export]{adjustbox}
\usepackage{subfig}
\usepackage{amssymb}
\usepackage{varwidth}
\usepackage{xspace}

\newcommand{\rknn}{R$k$NN\xspace}
\newcommand{\rstknn}{RST$k$NN\xspace}

\begin{document}

\markboth{}{Efficient Reverse $k$ Nearest Neighbor evaluation for hierarchical index}

\title{Efficient Reverse $k$ Nearest Neighbor evaluation for hierarchical index}
\author{Siddharth Dawar
\affil{Indraprastha Institute of Information Technology, India}
Vikram Goyal
\affil{Indraprastha Institute of Information Technology, India}
Debajyoti Bera
\affil{Indraprastha Institute of Information Technology, India}
}

\begin{abstract}
``Reverse Nearest Neighbor'' query finds applications in decision support
systems, profile-based marketing, emergency services etc.  In this paper,
we point out a few flaws in the branch and bound algorithms proposed
earlier for computing monochromatic \rknn queries over data points stored
in hierarchical index. We give suitable counter examples to validate our claims and propose a
correct algorithm for the corresponding problem.
We show that our algorithm is correct by identifying necessary
conditions behind correctness of algorithms for this problem.
\end{abstract}

\maketitle

\section{Introduction}\label{section:intro}
One important type of operation that is gaining popularity in database and
data-mining research
community is the {\em Reverse Nearest Neighbor Query} (\rknn)
\cite{korn2000influence}. Given a set of database objects $O$ and a query
object $Q$, the \rknn query returns those objects in $O$, for which $Q$
is one of their $k$ nearest neighbors; here the notion of neighborhood is with
respect to an appropriately defined notion of distance between the objects.
A classic example \rknn is in the domain of decision support systems where the task is to open
a new facility (like a restaurant) in an area such that it will be least influenced
by its competitors and attract good business. Another application is profile
based marketing \cite{korn2000influence}, where a company maintains profiles of
its customers and wants to start a new service which can attract the maximum
number of customers.
\rknn has also applications in clustering, where a cluster could be created by
identifying a group of objects, and clustering them around their common nearest
neighbor point -- this essentially involves finding cluster centers
with high cardinality of reverse nearest neighbor sets.
{\em Reciprocal nearest neighborhood}, in which data
points which are nearest neighbors of each other are clustered
together (and therefore, satisfy both nearest neighbor and reverse nearest
neighbor criteria), is another well-known technique in clustering \cite{lopez2012fastrnn}.

This important concept has seen a series of remarkable applications and
algorithms for processing different types of
objects, in various contexts and under variations \cite{kang2007continuous}, \cite{safar2009voronoi}, \cite{tran2009reverse}, \cite{taniar2011spatial},\cite{shang2011finding}, \cite{cheema2012continuous}, \cite{ghaemi2012continuous},\cite{li2013efficient}, \cite{emrich2014reverse}, \cite{cabello2010facility}, \cite{bhattacharya2013new} of the problem parameters.
The focus of this paper is
monochromatic \rknn queries -- in this version, all objects in
the database and the query belong to the same category, unlike the bichromatic
version in which the objects can belong to different categories. Furthermore, we want to
focus on queries where $k$ is specified as part of a query, and want to support
objects from an arbitrary metric space.

This paper points out several fundamental inaccuracies in three papers published
earlier on the problem mentioned above.
\begin{itemize}
    \item Reverse $k$-nearest neighbor search in dynamic and general metric databases \cite{achtert2009reverse}
    \item Reverse spatial and textual $k$ nearest neighbor search \cite{lu2011reverse}
    \item Efficient algorithms and cost models for reverse spatial-keyword k-nearest neighbor search \cite{lu2014efficient}
\end{itemize}

Achtert et al.\cite{achtert2009reverse} proposed a branch-and-bound algorithm
for the above problem which could use any given hierarchical tree-like index
on data from any metric space. Lu et al. \cite{lu2011reverse} proposed a
similar algorithm, but specifically optimized for spatio-textual data, for answering
\rstknn queries using a specialized IUR tree as the
indexing structure. In a followup paper \cite{lu2014efficient}, they  proposed
an improvement of their algorithm (including correcting an error) and a
theoretical cost model to analyze the efficiency of their algorithm. However,
we observed several deficiencies in the algorithms mentioned above. In this
paper we will point out those inaccuracies, and discuss them more formally by
pointing out some key properties which
these algorithms violate, but are necessary for ensuring correctness of
these and other similar algorithms.  We will present detailed counter examples and suggest corrective
modifications to these algorithms.
Finally we will propose a correct algorithm for performing \rknn queries over a
hierarchical index and also present its proof of correctness.

The paper is organized as follows. In Section \ref{section:earlier} we explain
the three published approaches mentioned above in which we found
inaccuracies. In Section \ref{section:counter-ex} we describe our counter-examples with respect to them.
We present our modified algorithm in Section \ref{section:algo}, and its
proof of correctness in Section \ref{section:proof}.

\section{Earlier Results}\label{section:earlier}
The underlying algorithms for all three approaches mentioned above essentially
have the same structure and follow a branch-and-bound approach. The former work
is applicable on any kind of data with a distance measure that is a metric, and
uses any hierarchical tree-like index built on the data. The two latter work are
specifically concerned with \rknn query on spatio-textual data, which they
refer to as \rstknn query.

\begin{figure}[!tbh]
\begin{center}
\subfloat[$\text{}$]{\includegraphics[width=6.275cm]{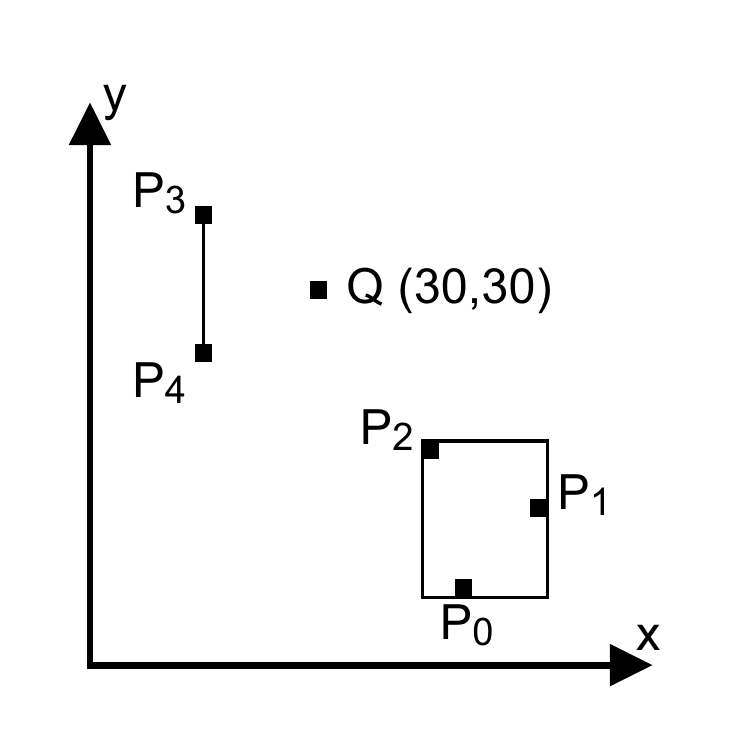}}
\subfloat[$\text{}$]{\includegraphics[width=6.275cm]{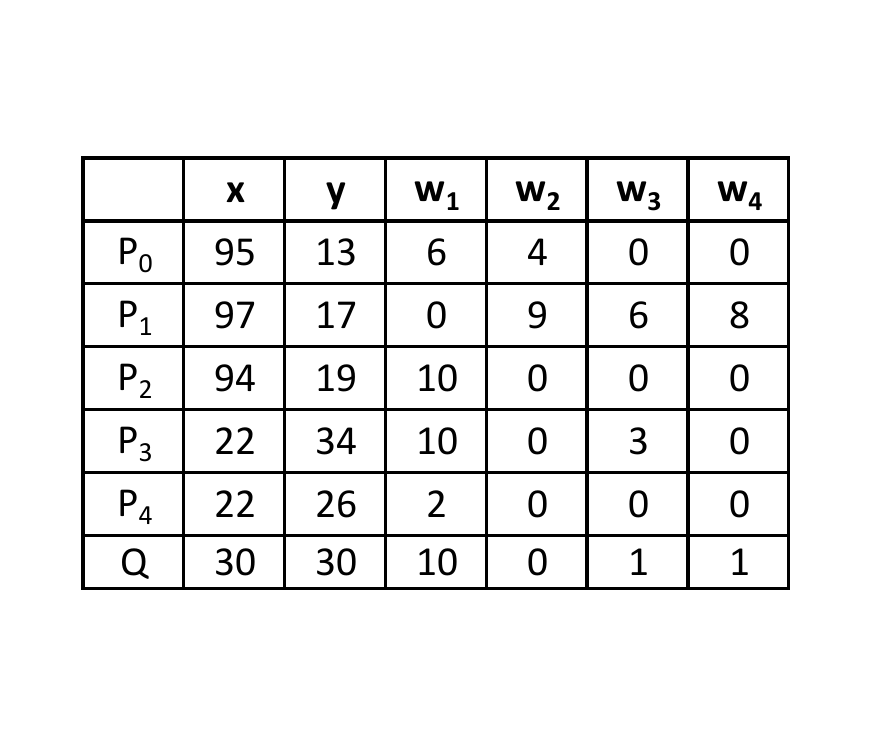}}
\caption{\small Example for illustrating \rstknn and \rknn}
\label{fig:example}
\end{center}
\end{figure}

In \rstknn, each object is
represented by a pair ($loc,vct$) where $loc$ is the spatial location and $vct$ is the
associated textual description which is represented by (word,weight(word)) pairs
for all words appearing in the database. Weight of a word
is calculated on the basis of TF-IDF scheme \cite{salton1988term}. Spatio-textual similarity ($SimST$) is defined by \cite{lu2011reverse} as follows:
\begin{equation}
SimST(o_1,o_2)=\alpha*(1-\frac{dist(o_1.loc,o_2.loc)-\varphi_s}{\psi_s- \varphi_s})+ (1-\alpha)*(\frac{EJ(o_1.vct,o_2.vct)-\varphi_t}{\psi_t- \varphi_t})
\end{equation} 
The parameter $\alpha$ is used to define the relevance factor for spatial and
textual similarity while calculating the total similarity scores and is
specified in a query. $\varphi_s$ and $\psi_s$ denote the minimum and maximum
distance between any two objects in the database and are used to normalize the
spatial similarity to the range $[0,1]$. Similarly $\varphi_t$ and $\psi_t$ denote
the minimum and maximum textual similarity between any two objects in the
database. $dist(\cdot)$ is the Euclidean Distance between $o_1$ and $o_2$ and
$EJ$ is the
Extended Jaccard Similarity \cite{tan2011v} defined as:
\begin{equation}
EJ(o_1.vct,o_2.vct)=\frac{ \sum_{j=1}^n o_1.w_j*o_2.{w'}_{j}}{\sum_{j=1}^n o_1.w_j^2+\sum_{j=1}^n o_2.{w'}_{j}^2-\sum_{j=1}^n o_1.w_j*o_2.{w'}_{j} }
\end{equation}
where $o_1$.vct=$\langle w_1,\ldots,w_n\rangle$ and $o_2$.vct=$\langle w'_1,\ldots,w'_n\rangle$.

As an example, consider Figure \ref{fig:example}. There, considering only
location attributes, and for $k=2$, \rknn of $Q$ are objects $P_3$ and $P_4$.
However, if we consider both spatial and textual similarity, and taking $k=2$
and $\alpha=0.4$, \rstknn of $Q$ is $P_2$, $P_3$ and $P_4$.

Now we will describe the actual algorithm proposed by \cite{lu2011reverse} for 
\rstknn. It is important to present it in some detail -- this is required for
proper appreciation of the inaccuracies in this algorithm. This algorithm
requires its data to be organized as an hierarchical index called as
IUR-tree. IUR-Tree is a R-Tree \cite{guttman1984r}; where every node of the tree is embedded with Intersection and Union Vectors. The textual vectors contain the weight of every distinct item in the documents contained in the node. The weight of every item in the Intersection Vector (resp. Union Vector) is the minimum weight (resp. maximum weight) of all the items present in the documents contained in the node. During the execution of the algorithm, a lower and upper nearest-neighbor list/contribution list is created and maintained for each node in the IUR-Tree. The lower (resp. upper) contribution list stores the minimum (resp. maximum) similarity between the node and its neighbors.
\begin{figure}[!tbh]
\begin{center}
\subfloat[$\text{}$]{\includegraphics[width=6.275cm]{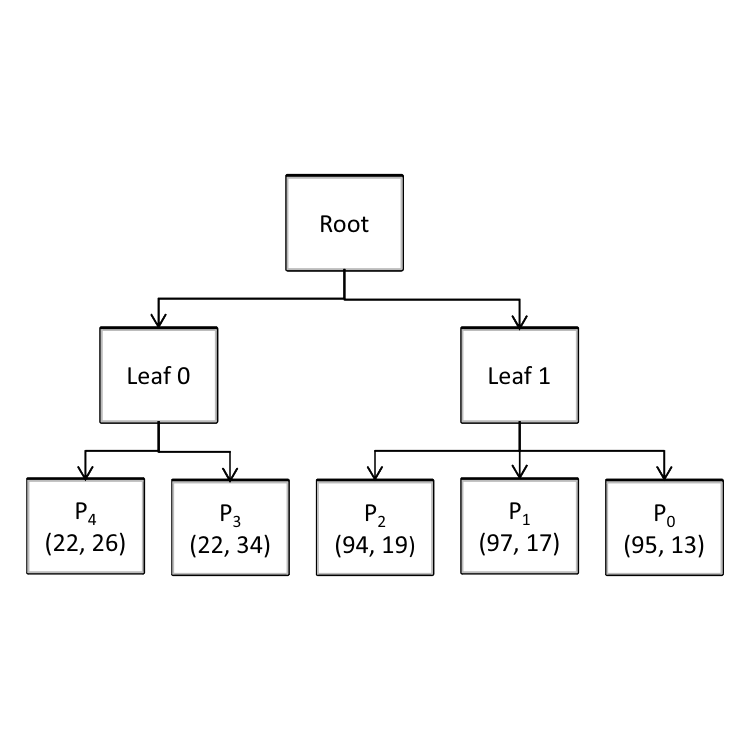}}
\subfloat[$\text{}$]{\includegraphics[width=6.275cm]{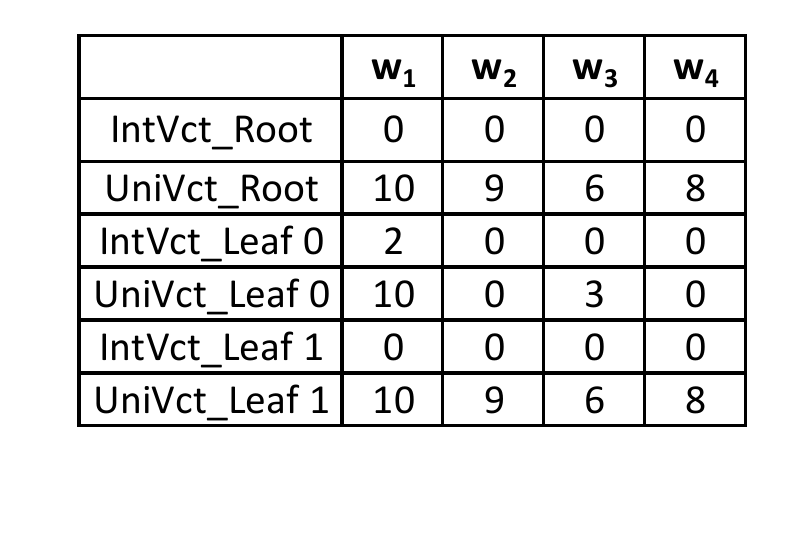}}

\caption{\small IUR-Tree and Textual Vectors of Fig 1 \label{fig:IUR}}
\end{center}

\end{figure}

The IUR-Tree and Intersection and Union Vectors of the corresponding nodes is shown in the Figure \ref{fig:IUR}.
 These vectors along with the MBR's of nodes are used to compute the similarity approximations i.e. upper and lower bounds on the spatio-textual similarity between two groups of objects.

\begin{spacing}{0.8}
\begin{algorithm}[tp]
\caption{\rstknn($R$: IUR-Tree root,$Q$: query) from \protect\cite{lu2011reverse}}
\begin{algorithmic}[1]
\State {\bfseries Output:} All objects $o$, s.t $o$ $\in$\rstknn($Q$,$k$,$R$).
\State Initialize a priority queue $U$, and lists $COL, ROL, PEL$;
\State EnQueue($U,R$);
\While{$U$ is not empty}
	\State $P \leftarrow$ DeQueue($U$); //Priority of $U$ is $MaxST(P,Q)$
	\For {each child node $E$ of $P$}
		\State Inherit($E.CLs, P.CLs$);
		\If{IsHitOrDrop($E,Q$)==false}
			\For{each node $E'$ in $COL,ROL,U$} //see subsection \ref{subsec:locality}
			\State UpdateCL($E,E'$); //update contribution lists of $E$;
			\If{IsHitOrDrop($E,Q$)=true} //see subsection \ref{subsec:completeness}
				\State{\bfseries break;}
			\EndIf
			\If{$E' \in U\cup COL$}
				\State UpdateCL($E',E$); //Update contribution Lists of $E'$ using $E$.
				\If{IsHitOrDrop($E',Q$)=={\bfseries true}}
					\State Remove $E'$ from $U$ or $COL$;
				\EndIf
			\EndIf
			\If{$E$ is not a hit or drop}
				\If{$E$ is an index node}
					\State EnQueue($U,E$);
				\Else
					\State $COL$.append($E$); //a database object
				\EndIf
			\EndIf
			\EndFor
		\EndIf
	\EndFor
\EndWhile
\State Final$\_$Verification($COL,PEL,Q$);
\algstore{myalg}
\end{algorithmic}
\end{algorithm}
\end{spacing}

We refer to an internal node or a point in the IUR-Tree as an entry. The algorithm takes as an input an IUR-Tree (Intersection Union tree) $R$,
query $Q$ and returns all database objects which are \rstknn of $Q$. The data structures used are: a priority queue ($U$) sorted
in decreasing order on $MaxST(E,Q)$, result list ($ROL$), pruned list ($PEL$) and
candidate list ($COL$). $MaxST(E,Q)$ is the maximum spatial textual similarity of the entry $E$ with the query point $Q$. The
algorithm dequeues the root of the IUR-Tree from the queue and for every child $E$ of the root,  inherits the contribution list of its parent. The function UpdateCL($E,E'$) is invoked and the contribution list of $E$ is updated with every $E'$ present in the candidate list, result list and the priority queue. After every invocation to UpdateCL(.), the algorithm checks based on the minimum and maximum bound similarity scores with the $k^{th}$
nearest neighbor, whether to add $E$ to the results,
candidates or pruned list. If $E$ can't be pruned or added to the results, the contribution list of $E'$ is updated with $E$. This process is called the mutual effect. If $E'$ can be added to the results or pruned, it is removed from the queue or $COL$. After updating node $E$ with all entries of $COL$, $ROL$ or $U$, the function \textit{IsHitorDrop()} is again invoked. If $E$ can't be added to the result or pruned list, a check is performed to find out whether $E$ is a internal node or a point. If $E$ is an
internal node, it is added to the queue, else to the candidate list.
When the queue becomes empty, there might be some objects left in the
candidate list. The function \textit{Final$\_$Verification()} is invoked where the candidate objects are updated with all the entries present in $PEL$ to decide whether they
belong to result or not.
\begin{spacing}{0.8}
\begin{algorithm}[tp]
\begin{algorithmic}[1]
\algrestore{myalg}
\Function{Final$\_$Verification($COL,PEL,Q$)}{}
\While{$COL \neq \emptyset$}
	\State Let $E$ be an entry in $PEL$ with the lowest level;
	\State $PEL=PEL-\lbrace E \rbrace$;
	\For {each object $o$ in $COL$}
		\State UpdateCL($o,E$); //update contribution lists of $o$.
		\If{IsHitOrDrop($o,Q$)==true} // see subsection \ref{subsec:completeness}
			\State $COL=COL-\lbrace o \rbrace$;
		\EndIf
	\EndFor
	\For{each child node $E'$ of $E$}
		\State $PEL=PEL \cup \lbrace E' \rbrace$; //access the children of $E'$ 	
	\EndFor
\EndWhile
\EndFunction
\end{algorithmic}
\end{algorithm}
\end{spacing}

\section{Counter-examples}\label{section:counter-ex}
We describe three counter example in this section:
\begin{enumerate}
    \item Inaccuracy regarding computation of $MinT$ and $MaxT$
    \item Inaccuracy w.r.t. Locality Condition
    \item Inaccuracy w.r.t. Completeness Condition
\end{enumerate}
All these examples are
illustrated with respect to the algorithm described in \cite{lu2011reverse};
however we also
explain the concepts used in constructing these examples -- therefore these
examples can be easily modified to suit the other algorithms.
We observed that \cite{achtert2009reverse} proposed an algorithm which maintains
the locality condition, but violates the completeness condition. We recently
observed that \cite{lu2014efficient} modified their previous algorithm from
\cite{lu2011reverse} which now maintains the locality condition. However, their algorithm still violates the completeness condition.

\subsection{Inaccuracy regarding computation of $MinT$ and $MaxT$}
The branch-and-bound algorithm presented in \cite{lu2011reverse} required
cleverly constructed lower and upper bounds on the textual similarity (and combined
textual-spatial similarity) between two groups of data objects. Its authors defined
$MinT$ (minimum possible similarity) and $MaxT$ (maximum possible similarity)
and claimed that these definitions, when used in conjunction with upper and
lower bounds on spatial similarity, give valid upper and lower bounds on the
similarity between two groups of objects. To prove this claim, they used the
following crucial lemma. The first inaccuracy we report is regarding this lemma.

\begin{definition}[Similarity Preserving Function] \cite{lu2011reverse} Given
    two functions $fsim: V \times V \rightarrow \mathbb{R}$ and $fdim:
    \mathbb{R} \times \mathbb{R} \rightarrow \mathbb{R}$, where $V$ denotes the
    domain of $n$-element vectors and $\mathbb{R}$, the real numbers. $fsim$ is
    a similarity preserving function w.r.t $fdim$, such that for any three vectors
    $\vec{p}= \langle x_1\ldots,x_n \rangle$, $\vec{p'} =\langle
    {x_1}^{'}\ldots,{x_n}^{'} \rangle$,$\vec{p^{''}}= \langle
    {x_1}^{''},\ldots{x_n}^{''} \rangle$, if $\forall i  \in [1,n]$,
    $fdim(x_i,{x_i}^{'}) \geq fdim(x_i,{x_i}^{''})$, then we have $fsim(\vec{p},\vec{p^{'}})\geq fsim(\vec{p},\vec{{p}^{''}})$.
\label{defn:similarity_PF}
\end{definition}

\begin{lemma}\cite{lu2011reverse}
    \label{lemma:ej-wrong-lemma}
    Extended Jaccard is similarity preserving function wrt. function $fdim(x,x') =\frac{min (x,x')}{max(x,x')}$ for $x,x'> 0$.
\end{lemma}


\paragraph{Counter Example}
Consider three points $p$, $p'$, $p''$ with textual vectors $\vec{p}= \langle 100,30 \rangle$, $\vec{p'}= \langle 1,40 \rangle$, $\vec{p''}= \langle 1,50 \rangle$.
Using $fdim(\cdot,\cdot)$ as defined in Lemma \ref{lemma:ej-wrong-lemma},
observe that the given points satisfy 
the conditions for a similarity preserving
function, i.e., $\forall i \in [1,2]$, $\frac{min(x_i,x'_i)}{max(x_i,x'_i)} \ge
\frac{min(x_i,x''_i)}{max(x_i,x''_i)}$. 
However, $EJ(p,p') = 0.116 \not\geq EJ(p,p'') = 0.135 $ which contradicts Definition
\ref{defn:similarity_PF}. The $MinT$ and $MaxT$ formula given in the paper relied
on the above Lemma to be correct, which therefore become invalid.


We now present our approach to calculate $MinT$ and $MaxT$ between two groups of
textual objects $E$ and $E'$.
As explained earlier, every textual object is represented as a vector of term
frequencies. For any group of objects, their intersection vector (resp. union
vector) has been defined to be a vector whose every coordinate is the minimum (resp. maximum)
frequency among the corresponding coordinates of objects. Denoting the
intersection and union vectors of $E$ as $\langle E.i_1, E.i_2, \ldots \rangle$
and $\langle E.u_1, E.u_2, \ldots \rangle$, notice that for every $o \in E$, and
$j \in [1,n]$, $E.i_j \le o.w_j \le E.u_j$. 
We propose the following formul\ae\ for $MinT$.
\begin{equation}
MinT(E,E')=\frac{ \sum_{j=1}^n E.i_j*E'.i_j}{\sum_{j=1}^n E.u_j^2+\sum_{j=1}^n E'.u_j^2-\sum_{j=1}^n E.i_j*E'i_j }
\end{equation}
The idea for computing $MinT$ is that since it is a lower bound, we want to
minimize the term in the numerator and maximize the denominator of EJ to ensure
that $\forall o \in E$ and $\forall o' \in E', EJ(o,o') \geq MinT(E,E')$.
Similarly formul\ae\ for $MaxT$ is given below:
\begin{equation}
MaxT(E,E')=\frac{ \sum_{j=1}^n E.u_j*E'.u_j}{\sum_{j=1}^n E.i_j^2+\sum_{j=1}^n E'.i_j^2-\sum_{j=1}^n E.u_j*E'u_j }
\end{equation}

\subsection{Inaccuracy w.r.t. Locality Condition}\label{subsec:locality}
\begin{figure}[!htb]
\begin{center}
\subfloat[$\text{Distribution of Points}$]
{\includegraphics[width=0.4\linewidth]{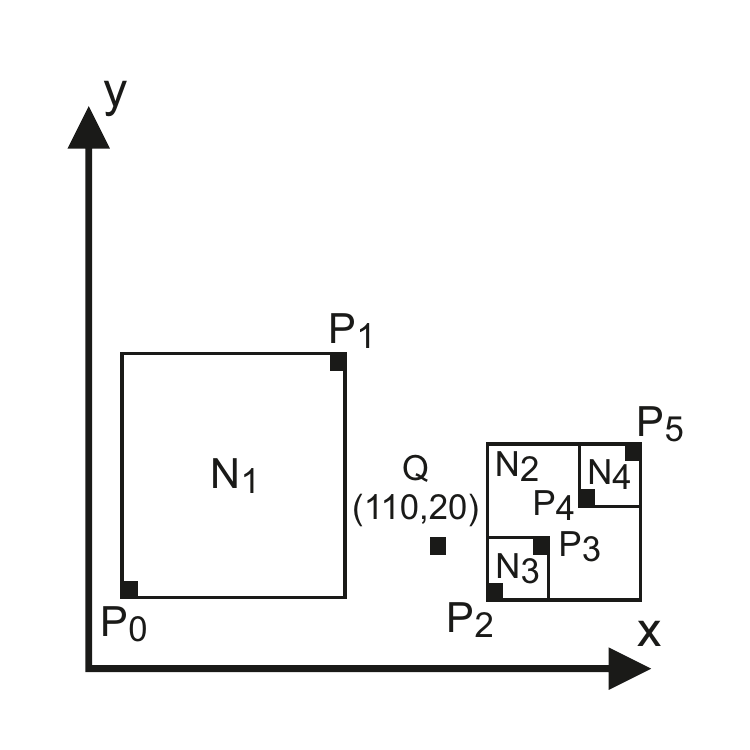}}
\subfloat[$\text{IUR Tree}$]{\includegraphics[width=0.6\linewidth]{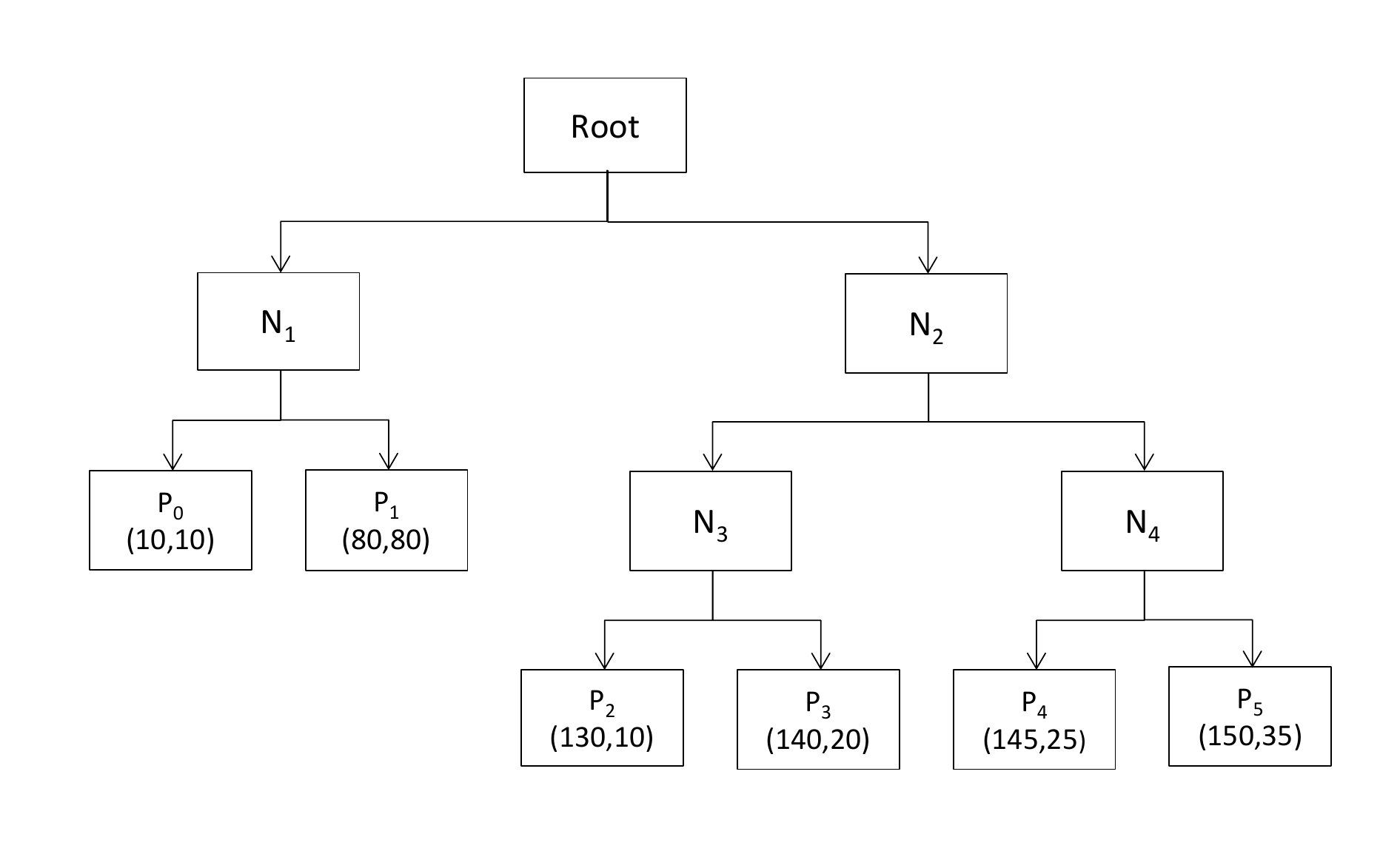}}
\caption{\small Counter-example (Locality and Completeness conditions)
\label{fig:counter-ex}
 } 
\end{center}
\end{figure}

Consider the following counter-example for the dataset and IUR-Tree illustrated in Figure
\ref{fig:counter-ex}, and let $\alpha=1$ and $k=2$. The minimum and maximum distance between any two points in the database is $\varphi_s$=7.07 and $\psi_s$=142.21. The exact \rstknn of the
query point $Q$ is $P_0$ and $P_1$. The trace of the algorithm \cite{lu2011reverse} is shown in Table
\ref{Tab:rstknn_algo}. We will focus on step 1 here. The root of the tree is dequeued from the tree and node $N_1$ is processed. $N_1$ inherits the contribution lists of its parent, which is empty. Since $U$, $ROL$ and $COL$ are empty, $N_1$ is simply added to the queue. Now, node $N_2$ is processed. $N_2$ updates
its upper and lower contribution lists with $N_1$ and invokes IsHitOrDrop. The upper and lower contribution lists of $N_2$ upon invoking IsHitOrDrop is :\\
$N_2.^{L}$.CL=$\lbrace (N_1,0,2) \rbrace$\\
$N_2.^{U}$.CL=$\lbrace (N_1,0.68,2) \rbrace$\\
Since $MinST(N_2,Q)=0.73$, which is more than the upper bound given by
$N_2.^{U}$.CL, at this point node $N_2$ is accepted (wrongly) as the \rstknn of $Q$. \begin{table}[!htb]
\begin{tabular}{cc}%
\noindent\begin{minipage}[t]{\linewidth}
\begin{center}
\caption{Trace of \rstknn Algorithm (2011) \label{Tab:rstknn_algo}}
    \begin{tabular}{|l | p{4cm} | l | l | l | l |}
    \hline
    \bfseries{Steps} & \bfseries{Actions} & \bfseries{U}& \textbf{COL} & \textbf{ROL} & \textbf{PEL} \\
    \hline
    1 & Dequeue Root, Enqueue $N_1$ & $N_1$ &$\emptyset$& $P_2$, $P_3$, $P_4$, $P_5$ & $\emptyset$\\ 
    \hline
    2 & Dequeue $N_1$ & $\emptyset$ & $\emptyset$ & $P_0$,$P_1$,$P_2$, $P_3$ $P_4$, $P_5$ & $\emptyset$\\
    \hline
    \end{tabular}
\end{center}
\end{minipage}
 \end{tabular}
\end{table} 

We attribute this fault to the violation of the {\em Locality
Condition}, a property that, we claim, must have been followed by these
algorithms.

\paragraph{Locality Condition} Nearest neighbors of data points in a node may belong to
the node itself; hence, every node
should compute similarity with itself and include itself as a candidate (along
with other similar nodes) in any test to prune or accept the node as \rstknn of $Q$.

In the counter-example above, node $N_2$ does not satisfy this condition
since its contribution lists do not contain itself or points inside it.

\subsection{Inaccuracy w.r.t. Completeness Condition}\label{subsec:completeness}
The trace of the algorithm \cite{lu2014efficient} is shown in Table \ref{Tab:rstknn_algo_2014}.
\begin{table}[!htb]
	\begin{tabular}{cc}%
		\noindent\begin{minipage}[t]{\linewidth}
			\begin{center}
				\caption{Trace of \rstknn Algorithm (2014) \label{Tab:rstknn_algo_2014}}
				\begin{tabular}{|l | p{4cm} | l | l | l | l |}
					\hline
					\bfseries{Steps} & \bfseries{Actions} & \bfseries{U}& \textbf{COL} & \textbf{ROL} & \textbf{PEL} \\
					\hline
					1 & Dequeue Root, Enqueue $N_1$, Enqueue $N_2$ & $N_1$, $N_2$ &$\emptyset$& $\emptyset$ & $\emptyset$\\ 
					\hline
					2 & Dequeue $N_2$ & $N_1$ & $\emptyset$ & $P_2$, $P_3$ & $N_4$\\
					\hline
					3 & Dequeue $N_1$ & $\emptyset$ & $\emptyset$ & $P_0$, $P_1$, $P_2$, $P_3$ & $N_4$\\
					\hline
				\end{tabular}
			\end{center}
		\end{minipage}
	\end{tabular}
\end{table} 

We will now focus on Step 2, when node $N_2$ is dequeued from the priority queue,
and its children are now being processed.
Node $N_3$ is now processed and it inherits the contribution lists of its parent $N_2$. The function IsHitOrDrop is called, but $N_3$  can't be pruned or added to the results. After invocation of IsHitOrDrop, $N_3$ updates its contribution list with itself to maintain the locality condition. $N_3$ further updates its contribution list with other entries present in $COL$, $ROL$ and $U$ sorted in the decreasing order of the maximum spatio-textual similarity with $N_3$. The upper and lower contribution list of $N_3$ is shown below :\\ \\
$N_3.^{L}$.CL=$\lbrace (N_3,0.94,1),(N_1,0,2) \rbrace$\\
$N_3.^{U}$.CL=$\lbrace (N_3,1,1),(N_1, 0.68,2) \rbrace$\\
Since $MaxST(N_3,Q)=0.90$, which is less than 0.68; so at this point 
$N_3$ is accepted (wrongly) as \rstknn of $Q$.
We claim that this faulty
behaviour is due to not ensuring the {\em Completeness Condition}, viz.,
absence of $N_4$ in contribution lists of $N_3$.
This condition is discussed in more detail in Section \ref{sec:nn-u}. In this
example, the contribution lists of $N_3$ is not complete.



\section{Proposed \rstknn Query Algorithm }\label{section:algo}
In this section, we present a modified algorithm to answer \rstknn
queries. We will illustrate our algorithm with an example, pointing
out the modifications and end this sections with a formal
proof of correctness. We begin by formalizing some notions which will be
used in the algorithm, and will be crucial in ensuring its correctness.

As explained earlier, the algorithms we considered worked on data that was stored in a hierarchical
tree-like index, where the leaf nodes are data points themselves (to
be represented by small letters) and internal nodes (to be represented by
CAPITAL letters) contain pointers to children
nodes.
Our modified algorithm will share backbone of these algorithms; however,
structually, it will bear resemblance to the algorithm presented in
\cite{lu2011reverse,lu2014efficient}. However, it will be presented in a
generalized manner which can be used to perform \rknn queries, given any value
of $k$, on a wide variety of data and 
independent of the explicit indexing structure used.
The only requirement from the data and the index is a similarity measure $Sim(\cdot,\cdot)$
among the data points, information about the of number of objects in each
node and estimates $MinSim$ and $MaxSim$ among nodes (explained below).

\subsection{Contribution List a.k.a. $NN$-list}
We will use the following notation: if $e'$ is the $k^{th}$ nearest neighbor of $e$, then we will write $e'$ as $kNN(e)$.
We will use the convention that
a point is the $0^{th}$ nearest neighbor of itself.
An immediate observation is the following: $Sim(e,kNN(e)) \geq Sim(e,k'NN(e))$
for any $k' \geq k$.

One way to answer \rknn queries
is by computing the list of nearest neighbors ($NN$-list) for every data point $e$:
$NN(e)$ is an ordered list of data points $\langle e_1,e_2,e_3,\ldots\rangle$ such
that $e_1$ is $1NN(e),$ $e_2$ is $2NN(e)$ and so on. Computing this list explicitly for
every data point could be very inefficient. 
The usual approach followed by branch-and-bound algorithms like
\cite{achtert2009reverse,lu2011reverse,lu2014efficient} is searching the index top-down
while maintaining two NN-lists with each node - one contains an
overestimate of its nearest neighbor, and another containing an underestimate of
the same. These estimated lists are constructed using two functions
$MinSim(\cdot, \cdot)$ and $MaxSim(\cdot, \cdot)$ which must satisfy the
property below. The actual implementation of these functions depend crucially on
the type of data used and the index. For two nodes $E$ and $E'$,
\begin{itemize}
    \item $MinSim(E,E')$ must give a lower bound for the minimum similarity between pairs of
points from $E$ and $E'$ i.e. $\forall e \in E$ ,$ \forall e' \in E'$,
$Sim(e,e') \geq MinSim(E,E')$.
    \item $MaxSim(E,E')$ must give an upper bound for the maximum similarity between pairs of points from $E$ and $E'$ i.e. $\forall e \in E$ , $\forall e' \in E'$, $Sim(e,e') \leq MaxSim(E,E')$.
\end{itemize}

Next, we will define the main component of our algorithm, a formalization of
{\em contribution lists} (CL) used in earlier algorithms.

\begin{definition}[NN-list \label{NN-list}]
An NN-list of a node $E$ is a list of tuples: $\langle (E_1,m_1),(E_2,m_2)
\ldots \rangle$, where each $E_i$ is a node and $m_i$ is a positive integer.
\end{definition}

The NN-lists we will maintain per node are $NN_U$($E$) and $NN_L$($E$) whose tuples will
provide estimates to the similarity of $E$ to its $r^{th}$ nearest neighbor, for various
values of $r$.

\subsection{Lower bound list $NN_L$}
The central idea behind the $NN_L$ list comes from the following observation.
Suppose for a set of $m$ points $\lbrace e_1',e_2',\ldots,e_m' \rbrace$ and
another point $e$, we have that $Sim(e,e_i') \geq s$. Then, it is obvious that
if $e$ does not belong to this set, $Sim(e,mNN(e)) \geq s$; and if $e$ belongs
to this set, then $Sim(e,(m-1)NN(e)) \geq s$. Extending this concept to nodes,
consider any node $E$ with $m$ data points; now, if $MinSim(E,e) \geq s$ then,
$Sim(e,mNN(e)) \geq s$ if $e \not  \in  E$  and $Sim(e,(m-1)NN(e)) \geq s$ if $e
\in  E$. Notice that these bounds are tight.

We can even extend this idea to multiple nodes to get the following claim. Let
$e$ be a data point and $E_1,E_2,\ldots,E_k$ be a collection of non-overlapping
nodes which do not contain $e$, where the list is sorted in decreasing order of
$MinSim(E_i,e)$. Let $m_i$ denote the number of data points in $E_i$, and let
$s_i$ be a lower bound on $MinSim(E_i,e)$. Then, for all $j= 1\ldots k, Sim(e,
(\sum_{i=1}^j) m_i)NN(e)) \geq s_j$. If $e \in E_i$ for some $i$, then $m_i$
must be replaced with $m_i-1$. We can generalize this even further by
considering a node instead of $e$.

\begin{definition}[Lower NN-list \label{lower-contri-list}]
An NN-list $\langle (E_1,m_1),\ldots \rangle$ of non-overlapping nodes is a
valid $NN_L(E)$ if:  
\begin{itemize}
\item the list is sorted in decreasing order of $MinSim(E_i,E)$
\item for all $e \in E$, if $E$ does not overlap with $E_i$, then $m_i \leq
    \vert E_i \vert$ and if $E$ overlaps with $E_i$, then $m_i \leq \vert E_i \vert - 1$
\end{itemize}
\end{definition}

The following lemma describes the use of lower NN-lists to get underestimates of
nearest neighbors. The proof is immediate from earlier definitions.
\begin{lemma}\label{lemma:NNL-lemma}
    For any $t$ and $i$ that satisfies $\sum_{k=1}^{i-1} m_k < t \leq
    \sum_{k=1}^i m_k$ (including the case $t \le m_1, i=1$), it holds that for all $e \in
    E$, $Sim(e,tNN(e)) \geq MinSim(e,E_i)$.
\end{lemma}

\subsection{Upper bound list $NN_U$}\label{sec:nn-u}
We want to define $NN_U$ as an overestimation of nearest neighbors similar to
$NN_L$ and derive a similar lemma as Lemma \ref{lemma:NNL-lemma}; however, we
require an additional concept first.

\begin{definition}[Complete NN-list]
    We say that an NN-list $NN(E)$ is {\em complete} if every data
    point is present in some node in the NN-list, and for every $(E_i,m_i)$ in the list, 
\begin{itemize}
\item if $E$ does not overlap with $E_i$, then $m_i = \vert E_i \vert$
\item if $E$ overlaps with $E_i$, then $m_i = \vert E_i \vert - 1$
\end{itemize}
\end{definition}

It must
be noted that an $NN_L$ list need not be complete for it to satisfy Lemma
\ref{lemma:NNL-lemma}. However,
similar arguments do not work for $NN_U$. Take for example, the example
situation similar to
the one described for $NN_L$: we have a set of points $\lbrace e_1',e_2',\ldots
e_m' \rbrace$ and another point $e$ (all distinct). But even if we know that
$Sim(e,e_i') \leq s$ for some $s$ and for all $i$, it is nevertheless not true
that $Sim(e,mNN(e)) \leq s$, unless, all points other than $e$ are in the set --
which is precisely what a complete NN-list specifies.

Now we can define similar concepts like $NN_L$.
\begin{definition}[Upper NN-list]
For a node $E$, an NN-list $\langle (E_1,m_1),\ldots \rangle$ of non-overlapping nodes is a valid $NN_U(E)$ when the following holds:
\begin{itemize}
\item the list is sorted in decreasing order of $MaxSim(E_i,E)$
\item the list is complete
\end{itemize}
\end{definition}

Observe that the completeness condition requires that $NN_U(E)$ must contain $E$
itself, or its parent node, or all its children nodes -- this is essentially
the locality condition we mentioned earlier (Section \ref{subsec:locality}).
However, we have chosen to specifically highlight the above condition separately
from the more general completeness condition.
The main working lemma for $NN_U$ follows next.
\begin{lemma}\label{lemma:NNU-lemma}
    For any $t$ and $i$ such that $\sum_{k=1}^{i-1} m_k < t \leq \sum_{k=1}^i
    m_k$ (including the case when $i=1$ and $t \le m_1$), it holds that for all $e \in
    E$, $Sim(e,tNN(e)) \le MaxSim(e,E_i)$.
\end{lemma}

\subsection{Branch-and-bound traversal}
A branch-and-bound algorithm traverses a hierarchical index by first visiting
the root, and then exploring its children nodes, and so on. For every node it
visits, the algorithm decides what to do next based on some estimate of the
relevance of the current node to the desired answer (here, $NN_U$ and $NN_L$
lists). It may choose to further explore the node, add all the points
in the node to the result set and not explore the node further (aka. {\em
accepting} the node), or, simply not explore the node further because it decided
that the node does not contain any point that should be in the result set (aka.
{\em pruning} the node).

Suppose the query point is denoted by $Q$; and suppose that a branch-and-bound
algorithm is currently visiting $E$ during its traversal of the
index. Let $NN_L(E)$ denote the (valid) lower NN-list of $E$
node, and $NN_U(E)$ denote its (valid) upper NN-list. Also, suppose $i$ is the smallest
index such that $k \le \sum_{t=1}^{i} m_k$ for $NN_L(E)$, and 
$j$ is the smallest similar index for $NN_U(E)$.

Here are the main theorems that give us sufficient conditions for accepting and
pruning certain nodes in the index during a branch-and-bound traversal.

\begin{theorem}[Accepting and Pruning
    Condition]\label{thm:accept-prune-condition}
\begin{enumerate}
    \item If $MaxSim(E,Q) \leq MinSim(E,E_i)$, then $Q$ cannot have any node in $E$ in
its \rknn set. Therefore, $E$ can be pruned.
    \item If $MinSim(E,Q) > MaxSim(E,E_j)$, then all nodes in $E$ belong to \rknn of $Q$
and so $E$ can be accepted.
\end{enumerate}
\end{theorem}
The proofs for the two cases are immediate from
Lemma \ref{lemma:NNL-lemma} and \ref{lemma:NNU-lemma}, respectively.
\footnote{For accepting or pruning, in case there is a tie between
similarities between query point and a database point, we tie-break in favour of points in the database. The alternative
approach requires straight forward modification to the results in this
subsection.}.

\subsection{Algorithm}
Now we will discuss the modified algorithm for finding reverse nearest neighbors on
spatial-textual objects. Our algorithm is a modification of the one proposed in
\cite{lu2011reverse}, so we will mostly engage in highlighting the major changes.
Like the original algorithm, our algorithm uses the following data structures: a FIFO queue ($U$), a result
list ($ROL$), candidate list ($COL$) and pruned list ($PEL$). We use a FIFO
queue instead of a priority queue, as each entry of needs to update its
NN-list with every other entry present in every list  in order to
ensure completeness of lists. So, the order in which other entries are added is
irrelevant. We will frequently use NN-lists to refer to both the upper and lower
NN-lists of the corresponding entry.

As before, the algorithm initializes the lists and enqueues the root of the IUR-tree.
While the queue is not empty, an entry $E$ is dequeued from the queue and its
parent is removed from its NN-list. The two key modifications we suggest are
stated next. First, if $E$ is an internal node of the tree, it
adds itself to its NN-lists, thereby maintaining the {\em locality condition}
(line 12).
Then $E$ updates its NN-lists with each entry $E'$ present in the queue and vice
versa. The updation of NN-list of $E$ with every other entry in the queue
maintains the {\em completeness condition} (line 14). After this, IsHitorDrop is invoked to check if $E$ can
be pruned or added to the results. If $E$ can neither be pruned nor added to the
results, its children are added to the queue if $E$ is an internal node;
otherwise, $E$ is added to the candidate list. We continue with the optimisation
of having the children of $E$ copy the NN-list of $E$ before they are enqueued to $U$.
When the queue
becomes empty, there might be some candidate points left in the candidate list.
The procedure Final$\_$Verification is invoked to decide whether the points
present in the candidate list belong to the result list or the pruned list; this
procedure essentially checks every candidate point with other entries.

\begin{spacing}{0.8}
\begin{algorithm}[tp]
\caption{\rstknn($R$: IUR-Tree root,$Q$: query)}
\begin{algorithmic}[1]
\State {\bfseries Output:} All objects $o$, s.t $o$ $\in$RST$k$NN($Q,k,R$).
\State Initialize a FIFO queue $U$, and lists $COL, ROL, PEL$;
\State EnQueue($U,R$);
\While{$U$ is not empty}
	\State $E$ $\leftarrow$ DeQueue($U$); //FIFO Queue
	\For{each tuple $\langle E'_i,num_i \rangle \in NN_L(E)$}
		\If{$E'_i=E$ or $E'_i=Parent(E)$}
			\State remove $\langle E'_i,num_i \rangle$ from
			$NN_L(E)$ and $NN_U(E)$ ;
		\EndIf
	\EndFor
	\If( $E$ is an internal node)
		\State Additself($E$)				//Ensure locality condition
	\EndIf	
	\For{each entry $E'$ in $U$}   // Ensure completeness condition   
		\State Update$\_$NN-list($E,E'$); //mutual effect
		\State Update$\_$NN-list($E',E$); //mutual effect
	\EndFor
	\If{$E$ is not a hit or drop}
		\If{$E$ is an index node}
			\For{each child $C_E$ of $E$}
				\State Inherit($NN_L(C_E),NN_L(E)$);
				\State Inherit($NN_U(C_E),NN_U(E)$);
				\State EnQueue($C_E$)				
			\EndFor
			
		\Else
			\State $COL$.append($E$); 
		\EndIf
	\EndIf
\EndWhile
\State Final$\_$Verification($COL,PEL,ROL,Q$);
\algstore{ouralg}
\end{algorithmic}
\end{algorithm}
\end{spacing} 
We illustrate the working of our algorithm on the example presented
earlier (Figure \ref{fig:counter-ex}) in
Table \ref{Tab:our_algo}. As expected, the algorithm now correctly returns $P_0$ and $P_1$
as the only points in \rstknn of $Q$.

\begin{spacing}{0.8}
\begin{algorithm}[tp]
\begin{algorithmic}[1]
\algrestore{ouralg}
\Function{Final$\_$Verification($COL,PEL,ROL,Q$)}{}
\State $PEL\, =\, SubTree(PEL)$
%
\While{$COL \not=\emptyset$}
	\For {each point $o$ in $COL$}
		\For{each point $r$ in $ROL$}
			\State Update$\_$NN-list($o,r$); 
		\EndFor
		\For(each point $p$ in $PEL$)
			\State Update$\_$NN-list($o,p$); 
		\EndFor
		\For(each point $c'$ in $COL -\lbrace o \rbrace$ )
			\State Update$\_$NN-list($o,c'$); 
		\EndFor				
		\If{IsHitOrDrop($o,Q$)==true}
			\State $COL=COL- \lbrace o \rbrace$;
		\EndIf
	\EndFor
\EndWhile
\EndFunction
\end{algorithmic}
\end{algorithm}
\end{spacing}

\begin{table}[tp]
\begin{tabular}{l}
\begin{minipage}[t]{\linewidth}
\begin{center}
\caption{Trace of our algorithm \label{Tab:our_algo}}
    \begin{tabular}{ |l| p{4cm} | l | l | l | p{2cm} |}
    \hline
    \bfseries{Steps} & \bfseries{Actions} & \bfseries{U}& \textbf{COL} & \textbf{ROL} & \textbf{PEL} \\ \hline
    1 & Dequeue Root, Enqueue $N_1$, Enqueue $N_2$ & $N_1$, $N_2$ &$\emptyset$&$\emptyset$&$\emptyset$\\ \hline
	2 & Dequeue $N_1$ & $N_2$, $P_0$, $P_1$ & $\emptyset$&$\emptyset$&$\emptyset$\\ \hline
	3 & Dequeue $N_2$ & $P_0$, $P_1$, $N_3$, $N_4$ &$\emptyset$&$\emptyset$ & $\emptyset$ \\ \hline
	4 & Dequeue $P_0$& $P_1$, $N_3$, $N_4$ &$\emptyset$&$P_0$& $\emptyset$\\ \hline
	5 & Dequeue $P_1$& $N_3$, $N_4$ &$\emptyset$&$P_0$, $P_1$& $\emptyset$\\ \hline
	6 & Dequeue $N_3$& $N_4$, $P_2$, $P_3$ &$\emptyset$&$P_0$, $P_1$& $\emptyset$\\ \hline
	7 & Dequeue $N_4$& $P_2$, $P_3$ &$\emptyset$&$P_0$, $P_1$& $N_4$\\ \hline
    8 & Dequeue $P_2$& $P_3$ &$P_2$&$P_0$, $P_1$& $N_4$\\ \hline
    9 & Dequeue $P_3$& $\emptyset$ &$P_2$&$P_0$, $P_1$& $N_4$, $P_3$\\ \hline
    10 & Verify $P_2$& $\emptyset$ &$\emptyset$&$P_0$, $P_1$& $N_4$, $P_3$, $P_2$\\ \hline
    \end{tabular}
\end{center}
\end{minipage}
 \end{tabular}
\end{table}

\subsection{Proof of Correctness}\label{section:proof}
We will now give a formal proof of correctness of our algorithm. Essentially, we
will show that, when an index node is checked (line 18) if it can be immediately accepted
or pruned (using Theorem \ref{thm:accept-prune-condition}),
its NN-lists (especially, upper NN-list) are complete (hence, valid).

First, we want to discuss a few observations. The first fact is, if at
any point of time, a
data point $e$ not belonging to an entry $E$ is covered in $NN(E)$, then
$e$ is covered subsequently in the NN-list of $E$.
Since $e$ is covered at this instant, some ancestor $E^*$ of $e$ must be present in
the NN-list of $E$ at that instant. Observe that after an entry is added to the NN-list of $E$,
it is removed from the  NN-list of $E$ only when the NN-list of $E$ is updated
with the children of $E^*$ (lines 21,22). This ensures that $e$ is forever covered in the
NN-list of $E$.

Similarly, $e$ is covered subsequently
in the NN-lists of all (sub-)children of $E$.
At line 18 of the algorithm, if $E$ can't be added to the results or pruned,
after updating its NN-list with each entry present in $U$,its children are added
to the queue. However, each child of $E$ inherits its NN-list i.e. simply copies
its NN-list (lines 21,22). Therefore, the children of $E$ will also have $e$ in
their NN-list.


\begin{figure}[tp]
\begin{center}
\subfloat[$\text{e is in sub tree of E}$]{\includegraphics[width=0.3\linewidth]{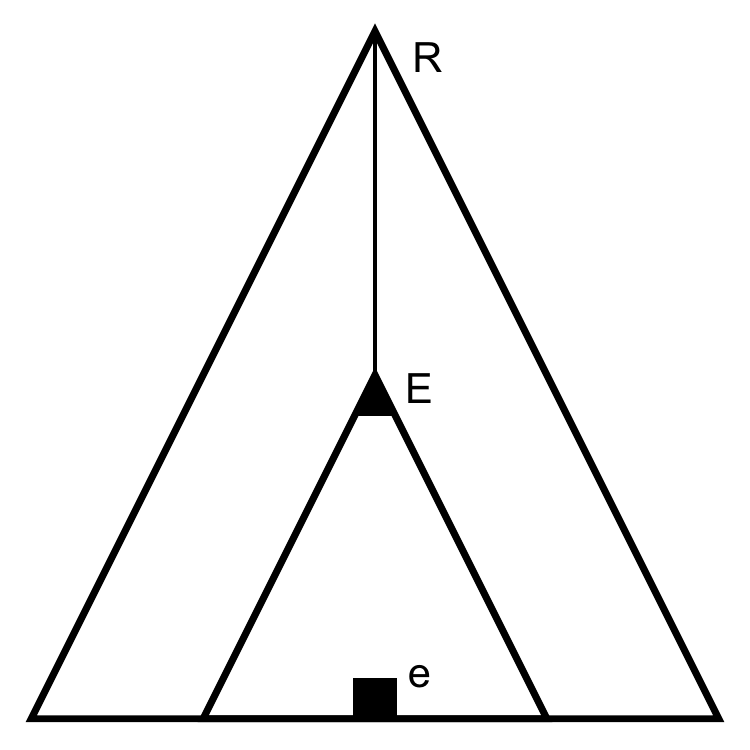}}
\subfloat[$\text{e is not in sub tree of
E}$]{\includegraphics[width=0.3\linewidth]{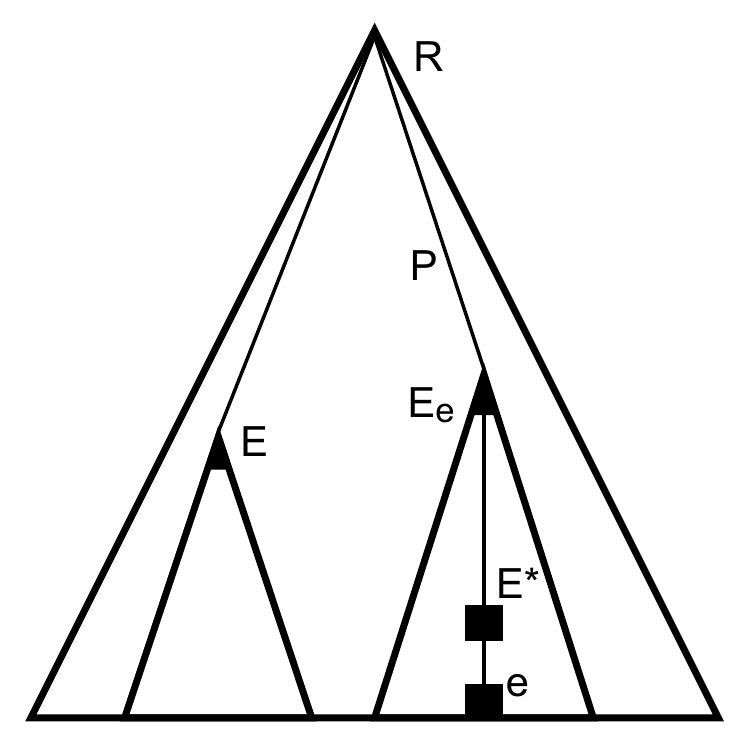}}
\caption{\small Indexing tree \label{fig:proof}}
\end{center}
\end{figure}

Now we present the key lemma for our proof of correctness.

\begin{lemma}\label{lemma:proof-3}
The upper NN-list of every entry $E$, which is dequeued from the queue, is complete after line 17 of the \rstknn algorithm.
\end{lemma}
\begin{proof}
Consider an execution of the algorithm, and suppose the current node to be
dequeued from the queue is denoted by $E$. Let $e$ be any data point and $P$
denote the path from root to $e$ in the tree.
We will prove that after line 17 of the
algorithm, $e$ belongs to $NN_U(E)$. 
There are two
possibilities (see Figure \ref{fig:proof} for reference):
\begin{description}
    \item [Case A] $e$ belongs to sub tree of $E$
    \item [Case B] $e$ does not belong to sub tree of $E$
\end{description}

Case A is trivial. If $e$ belongs to sub tree of $E$, it will be present in
$NN_U(E)$ after line 17, since any internal node adds itself to its NN-lists
(line 12).

Let us now consider Case B.
Let $t_1$ be the time when $E$ is dequeued from the queue.
Now, one of these four different possibilities must be true at $t_1$.
\begin{description}
    \item[Case B.1] Some node $E_e$ on the path $P$ belongs to the result list
	$ROL$.
    \item[Case B.2] Some node $E_e$ on the path $P$ belongs to the pruned list
	$PEL$.
    \item[Case B.3] Some node $E_e$ on the path $P$ belongs to the queue $Q$.
    \item[Case B.4] $e$ belongs to the candidate list $COL$.
\end{description}
\paragraph{Case B.1}
Let $t_0$ denote the time when line 17 was encountered after $E_e$ was dequeued. Once again, there are two
possibilities.
\begin{description}
    \item[Case: $E$ belongs to the queue at $t_0$] In this case, $NN_U(E)$ will contain
	$E_e$ through mutual effect (line 16) at $t_0$. This implies
	that $e$ is covered by $NN_U(E)$ at $t_1$.
    \item [Case: $E$ does not belong to the queue at $t_0$]
If $E$ does not belong to the queue, it implies that there exists some ancestor
of $E$, say $E^*$ (cannot be $E_e$ because of condition of Case B.1) which
belongs to the queue at time $t_0$. Then $NN_U(E^*)$ contains $E_e$ through
mutual effect (line 16). This implies that once
$NN_U(E^*)$ contains $e$, upper NN-lists of all its discendant nodes will also contain
$e$.
\end{description}
The proof for Case B.2 and Case B.3 is similar to Case B.1.

We now consider the remaining Case B.4. Since $e \in$ COL, it implies that
some ancestor $E^*$ of $e$ was dequeued from the queue prior to $t_1$. All the
node present in the queue then contained $E^*$ in their upper NN-list through
mutual effect. Therefore at $t_1$ $E$ contained $E^*$ in its
upper NN-list.
\end{proof}

\begin{theorem}
Given an integer $k$, a query point $Q$, and an index tree $R$, the algorithm 2 correctly returns all \rstknn points.
\end{theorem}
\begin{proof}
The correctness follows from the following observations that were made earlier.
\begin{itemize}
    \item Internal nodes are accepted or pruned (by IsHitOrDrop) only when the sufficient
conditions according to the Theorem
\ref{thm:accept-prune-condition} are met (using Lemma \ref{lemma:proof-3}).
    \item For the data points left in the candidate list $COL$, in Final$\_$Verification,
the (complete) NN-lists of every such point are updated with every other object
(present in candidate, result and pruned list), before IsHitOrDrop being called
on the point for directly accepting or pruning. Our Final$\_$Verification routine
implements this in a rather straight forward manner. In line 32 of this routine,
internal nodes present in $PEL$ are replaced with their contained points to
ensure that operations in this routine directly involve points.

\end{itemize}
\end{proof}

\section{Conclusion and Future Work}\label{sec:Conclusion and Future Work}
\rknn is an important problem in facility location, operations research,
clustering and other domains. We observed that a few published algorithms are
not fully correct. In this paper we presented a correct algorithm to compute \rknn on a
general data set organised as a tree. We first discussed counter-examples to
illustrate where the earlier algorithms made an error, and then discussed
the necessity of maintaining locality and completeness conditions for ensuring
the correctness of results. We finished by modifying one of the proposed
algorithms along with an explanation why our algorithm is correct.

In the future, we would like to extend our algorithm for performing bichromatic
\rstknn algorithm. We would further like to develop algorithms where the objects
are dynamic (e.g., moving in space, or textual attributes getting updated).

\bibliographystyle{ACM-Reference-Format-Journals}
\bibliography{todsssubmission}

\end{document}